\documentclass[a4paper]{amsart}
\usepackage{amsthm}
\usepackage{amsmath}
\usepackage{amssymb}
\usepackage{tikz}
\usepackage{color}
\usetikzlibrary{matrix}
\usepackage{dsfont}
\usepackage{graphicx} 
\usepackage[margin=1.2 in]{geometry}
\usepackage{stackrel}
\numberwithin{equation}{section}
\newcommand\norm[1]{\left\lVert#1\right\rVert}
\newcommand\normm[1]{\left|#1\right|}

\title{On steering in the C*-algebraic framework} 
\author{Micha{\l} Banacki}
\address{International Centre for Theory of Quantum Technologies (ICTQT), University of Gda\'{n}sk, Jana Ba\.{z}y\'{n}skiego 1A, 80-309 Gda\'{n}sk, Poland; Institute of Mathematics, Faculty of Mathematics, Physics and Informatics, Wita Stwosza 57, 80-308, Gda\'{n}sk, Poland; Institute for Theoretical Physics, University of Cologne, Höninger Weg 100, 50969, Cologne, Germany}
\email{michal.banacki@ug.edu.pl}

\theoremstyle{plain}
\newtheorem{theorem}{Theorem}[section]

\newtheorem{proposition}[theorem]{Proposition}
\newtheorem{corollary}[theorem]{Corollary}

\theoremstyle{definition}
\newtheorem{definition}[theorem]{Definition} 
 
\theoremstyle{remark}

\newcommand{\beq}{\begin{equation}}
\newcommand{\eeq}{\end{equation}}
\newcommand{\be}{\begin{eqnarray}}
\newcommand{\ee}{\end{eqnarray}}
\newcommand{\beg}{\begin{eqnarray*}}
\newcommand{\eeg}{\end{eqnarray*}}

\begin{document}

\begin{abstract}
We discuss a scenario of bipartite steering with local subsystems of the parties modeled by certain operator algebras. In particular, we formalize the notion of quantum assemblages in a commuting observables paradigm and focus on equivalent descriptions of such objects providing a systematic analysis of previously scattered approaches. We provide necessary and sufficient conditions for the equivalence of quantum commuting and tensor models that is stable under extensions of the trusted subsystem by arbitrary finite-dimensional ancillae. As a result, we show that the gap between two models of bipartite steering can be observed in an arbitrary scenario with two measurement settings
($m = 2$) and more than two outcomes ($k > 2$). We also demonstrate that the identified gap is independent of nonlocality arising from the refutation of Tsirelson’s conjecture. Finally, we provide no-go results concerning the possibility of post-quantum steering in this most general bipartite paradigm and discuss related corollaries regarding free probability and operator system approach as well as a link to Tsirelson's problem.
\end{abstract}
\maketitle
\section{Preliminaries}
Description of physical reality provided on the most fundamental level by the language of quantum mechanics and quantum field theory enables us to predict phenomena unprecedented within classical theories. In particular, the notion of quantum entanglement \cite{4H} gives rise to correlations stronger than those observed for classically correlated systems \cite{EPR,Bell64}.

Indeed, consider a Bell-type scenario with a joint quantum system shared by two separated parties A and B, where both of them perform local measurements choosing (at random) between measurement settings $x=1,\ldots, m$ or $y=1,\ldots, m$ and obtaining outcomes $a=1,\ldots, k$ or $b=1,\ldots, k$ respectively. A probabilistic description of such an experiment is given by a behaviour $P=\left\{p(ab|xy)\right\}_{a,b,x,y}$ formed as a collection of conditional probabilities $p(ab|xy)$ of obtaining outcomes $a,b$ when measuring $x,y$. Depending on the assumed paradigm one defines different sets of permitted correlations. We say that behaviour $P$ belongs to the set of quantum correlations ($P\in \mathcal{C}_{q}(m,k)$) if there exist Hilbert spaces $H_A, H_B$ with $\dim H_A,\dim H_B<\infty$, a unit vector $\psi\in H_A\otimes H_B$ and PVMs elements $P_{a|x}\in B(H_A)$, $Q_{b|y}\in B(H_B)$ such that $p(ab|xy)=\langle\psi,P_{a|x}\otimes Q_{b|y}\psi\rangle$. Similarly, we introduce the set of quantum spatial correlations $\mathcal{C}_{qs}(m,k)$ by relaxing the finite-dimensional conditions for $H_A, H_B$ and by $\overline{\mathcal{C}_{q}(m,k)}$ we denote the closure of the set $\mathcal{C}_{q}(m,k)$ (often described as $\mathcal{C}_{qa}(m,k)$).

Despite the fact that the evoked sets describe correlations obtainable within the principles of quantum mechanics (and in general beyond convex combinations of deterministic behaviours) they exclude a more refined framework inspired by the algebraic quantum field theory \cite{Haag, Haag0}, where one considers a net of C*-algebras (with the common unit) related to different regions of space-time. For that reason, it is necessary to introduce yet another set of correlations, namely we say that $P$ belongs to the set of quantum commuting correlations ($P\in \mathcal{C}_{qc}(m,k)$) if there exist a Hilbert space $H$, a unit vector $\psi\in H$ and PVMs elements $P_{a|x},Q_{b|y}\in B(H)$ such that $ [P_{a|x},Q_{b|y}]=0$ and $p(ab|xy)=\langle\psi,P_{a|x}Q_{b|y}\psi\rangle$. 

Finally, one may consider a bipartite Bell-type scenario described by correlations obeying only no-signaling constraints (without possible underlying quantum state shared by both parties). We say that $P$ belongs to the set of no-signaling correlations ($P\in \mathcal{C}_{ns}(m,k)$) if $\sum_ap(ab|xy)=\sum_ap(ab|x'y)$ for all $b,x,x', y $ and $\sum_bp(ab|xy)=\sum_bp(ab|xy')$ for all $a,x, y , y'$.

It is clear that within the above paradigms, we obtain the following chain of inclusions
$$\mathcal{C}_q(m,k)\subset \mathcal{C}_{qs}(m,k) \subset\overline{\mathcal{C}_{q}(m,k)}\subset 
\mathcal{C}_{qc}(m,k) \subset \mathcal{C}_{ns}(m,k).$$In general all the above inclusions are nontrivial \cite{Slofstra1, Slofstra2, Col, Col2, Prakash, Re}. In particular, the question concerning equality $\overline{\mathcal{C}_q(m,k)}=\mathcal{C}_{qc}(m,k) $ for all scenarios $(m,k)$ is known as the Tsirelson's conjecture \cite{Boris, Scholz} which is equivalent (due to \cite{Kirchberg, Fritz, Junge, Ozawa}) to the long-standing Connes' embedding problem \cite{Connes, Pisier} or so-called Kirchberg's conjecture \cite{ Kirchberg, Pisier} and plethora of other important questions in operator algebras \cite{Musat, Dykema, Harris1, Harris2, Kavruk}. Note that although the Tsirelson's conjecture has been refuted \cite{Re} both mathematical and physical consequences of this fact are far from obvious and still inspire new areas of investigation \cite{Goildbring1,Goldbring2, Gao, Cleve}.

Bell-type scenarios are not the only experimental setup witnessing non-classicality (nonlocality) of quantum description. The other important paradigm of a similar feature is given by the so-called steering scenarios. Here two separated parties still share a joint state of a global system but only one of the parties steers the local state of the other by local measurements with a random choice of settings $x=1,\ldots, m$ with possible outcomes $a=1,\ldots, k$ leading to the statistical description provided by the mathematical notion of assemblage.

Starting from the heuristic conception \cite{s35,S36} and recent reformulation \cite{WJD07} questions regarding quantum steering flourish into an active area of research \cite{R1, R2, steering1, WJD072, WJD073,  SNC04} including such generalizations as multipartite settings \cite{steering0} and post-quantum (but still no-signaling) cases \cite{SBCSV15, SAPHS18, HS18, SHSA20}. Since the phenomenon of steering is mostly studied from the perspective of quantum information theory (including foundational perspective, resource theory, and possible application for semi-device-independent cryptography) it is typically analyzed under the assumption that composed systems are modeled by the tensor product of Hilbert spaces (often with additional requirement of finite-dimensional spaces). In other words, most conclusions obtained within this paradigm may not be entirely suitable for considerations of Algebraic Quantum Field Theory \cite{Haag, Haag0}. There is a need for discussion on steering in a model of commuting observables acting on a single Hilbert space, similar to the previously evoked case of correlations. Although it has been shown that such an approach may lead to different conclusions than the tensor product formalism \cite{Nava, Yan} there is still a lack of a methodical approach to this topic that would not be based on interesting yet particular and ad-hoc examples. 

For that reason, motivated by previous results on sets of correlations \cite{Fritz, Junge, Ozawa} we propose a more systematic approach to steering in the C*-algebraic framework starting from the formalization of the notion of quantum commuting assemblages and ending with comparisons between these class of assemblages and classes obtained within the quantum tensor model or no-signaling approach.

\section{Steering in the finite-dimensional case}
Consider a scenario of two (space-like) separated, finite-dimensional subsystems A and B that share a common quantum state describing the physical properties of the joint system. Assume that one of the subsystems, connected to observer B, is fully trusted or characterized, while subsystem A is not (i.e. dimension of a local Hilbert space of subsystem A and a particular form of performed measurements are not known from the perspective of B). Assume that only the observer connected with subsystem A performs local measurements on the joint system choosing (at random) settings from $x=1,\ldots, m$ and obtaining outcomes $a=1,\ldots, k$ while the subsystem B of the second observer is given by $M_n(\mathbb{C})$. Let $p(a|x)$ denotes a probability of obtaining outcome $a$ while measuring $x$ on subsystem A. If a local state of subsystem B conditioned upon $a,x$ is given by a density matrix $\rho_{a|x}\in M_n(\mathbb{C})$, then probabilistic description of such steering scenario is mathematically given by the notion of assemblage, i.e. a collection of subnormalized states (positive operators) $\sigma_{a|x}=p(a|x)\rho_{a|x}\in M_n(\mathbb{C})$. We say that $\Sigma=\left\{\sigma_{a|x}\right\}_{a,x}$ is a quantum assemblage acting on $M_n(\mathbb{C})$ if there exists $M_m(\mathbb{C})$, a joint state (density matrix) $\rho_{AB}\in M_m(\mathbb{C})\otimes M_n(\mathbb{C})$ and POVMs (positive operator-valued measures) elements $M_{a|x}\in M_m(\mathbb{C})$ such that 
\begin{equation}\label{s1}
\sigma_{a|x}=\mathrm{Tr}_A(M_{a|x}\otimes\mathds{1}\rho_{AB})
\end{equation}for any measurement $x=1,\ldots, m$ and outcome $a=1,\ldots, k$. Note that in principle, one could permit the infinite-dimensional description of subsystem A while keeping the finite-dimensional character of subsystem $B$, however, as we see below (compare Theorem \ref{Gisin}) in the case of assemblages on $M_n(\mathbb{C})$ this does not change anything. Similarly, since for any finite set of POVMs elements $M_{a|x}\in B(H)$ there exists a common isometry $V:H\rightarrow \tilde{H}$ such that $M_{a|x}=V^*P_{a|x}V$ for some PVMs (projection-valued measures) elements $P_{a|x}$, one could equivalently define quantum assemblages using PVMs instead of POVMs.

Following the intuition coming from Bell-type experiments under general no-signaling principle for correlations \cite{PR}, one can relax the quantum mechanical restrictions and consider a generalized scenario where only characterized subsystem B is of quantum nature, while the composed system is described by some possibly post-quantum theory yet still bounded by no-signaling constraints. Here local measurements performed on subsystem A are described just by classical labels of inputs $x$ and outputs $a$. We define a no-signaling assemblage on $M_n(\mathbb{C})$ as a collection $\Sigma=\left\{\sigma_{a|x}\right\}_{a,x}$ of subnormalized states (positive operators) $\sigma_{a|x}\in M_n(\mathbb{C})$ such that for any fixed $x$ 
\begin{equation}\label{s2}
\sum_a\sigma_{a|x}=\sigma
\end{equation}for some fixed state (density matrix) $\sigma\in M_n(\mathbb{C})$ that can be interpreted as a local state of a quantum subsystem B (before the steering measurements occurred). Since all assemblages given as (\ref{s1}) fulfill (\ref{s2}), quantum assemblages form a subset of all no-signaling ones. However, in the bipartite case consideration of no-signaling approach does not lead to any new phenomena \cite{HJW93, Gisin} (see also \cite{SBCSV15, SHSA20} for explicit formulation of that result).

\begin{theorem}\label{Gisin}
Let $\Sigma=\left\{\sigma_{a|x}\right\}_{a,x}$ be a no-signaling assemblage on $M_n(\mathbb{C})$, then there exists a density matrix $\rho_{AB}\in M_m(\mathbb{C})\otimes M_n(\mathbb{C})$, and POVMs elements $M_{a|x}\in M_m(\mathbb{C})$ such that $\sigma_{a|x}=\mathrm{Tr}_A(M_{a|x}\otimes\mathds{1}\rho_{AB})$.
\end{theorem}

The proof of the above theorem is based on the direct construction of a certain quantum realization of $\Sigma=\left\{\sigma_{a|x}\right\}_{a,x}$ - in particular operators $\sigma_{a|x}$ provide building blocks for POVMs elements $M_{a|x}$. Interestingly, this no-go result does not extend to the tripartite case (with the addition of yet another separate untrusted subsystem on which the local observer performs measurements). In that case, there exists post-quantum steering, i.e. there is a gap between quantum and no-signaling description \cite{SBCSV15, SAPHS18, HS18}.

In the following considerations, we address the issue of similar comparison between quantum and no-signaling models (for the paradigmatic bipartite case) in a generalized framework of arbitrary (unital) C*-algebras.

\section{No-signaling, quantum commuting and quantum tensor assemblages}

Steering scenarios considered in the previous section, typically in the context of quantum informational tasks, assume that the local subsystem of a trusted party is characterized by an algebra of observables given as a full matrix algebra $M_n(\mathbb{C})$ (corresponding to some finite-dimensional Hilbert space). However, since arbitrary quantum systems may require descriptions based on infinite-dimensional Hilbert spaces, there is a need for refinement of presented concepts. For that reason, in the most general case,  local algebras of observables should be modeled within the framework of (unital) C*-algebras that can be concretely realized as norm closed $*$-subalgebras in the algebra $B(H)$ of bounded operators acting on some Hilbert space $H$.

In order to extend the previously discussed notion of an assemblage to this generalized setting, one should abandon the density matrix interpretation and consider instead a C*-algebraic notion of (abstract) state given as a positive (hence bounded) functional $\phi$ that is normalized, i.e. $\phi(\mathds{1})=1$ (in what follows we will denote a set of all states on a unital C*-algebra $B$ by $S(B)$). Note that in a case of matrix algebra $M_n(\mathbb{C})$ density operators $\rho\in M_n(\mathbb{C})$ are in bijective correspondence with states $\phi\in S(M_n(\mathbb{C}))$ given by the formula $\phi(\cdot)=\mathrm{Tr}(\rho\cdot)$. These observations enable us to consider previously discussed finite-dimensional examples as a particular case within the new formalism.

From now on we will restrict our attention to nontrivial scenarios with at least two choices of measurements settings ($m\geq 2$) and measurement outcomes ($k\geq 2$). For notational simplicity, as above, we will consider only a case when number of possible outcomes is the same for each choice of a measurement setting. All results stated below remain true in a generalized case when number of measurement outcomes $k_m$ may depend a on measurement setting $m$.

Following \cite{Fritz, Yan} we recall the definition of quantum assemblage, coming from local measurements on the tensor product of Hilbert spaces.

\begin{definition}\label{qt}
We say that a collection of positive functionals $\Sigma=\left\{\sigma_{a|x}\right\}_{a,x}$ on a unital C*-algebra $B$ is a quantum tensor assemblage if there exist a Hilbert space $H$ together with a state $\phi \in S(B(H)\otimes_{min}B)$ and PVMs elements $P_{a|x}\in B(H)$  such that $$\sigma_{a|x}(\cdot)=\phi(P_{a|x}\otimes \cdot)$$
for all $a=1,\ldots, k$, $x=1,\ldots, m$. For a given scenario $(m,k)$ we denote the set of all quantum tensor assemblages on $B$ by $\mathbf{qA}(m,k,B)$.
\end{definition}

Indeed, note that according to the construction of the minimal tensor products of C*-algebras, the above definition covers exactly those assemblages that can be obtained starting from a unital inclusion (given by a $*$-homomorphism) of $B\subset B(\tilde{H})$ and (abstract) states on $B(H\otimes \tilde{H})$. In other words, it recovers the usual quantum mechanical approach to composition of subsystems, where local algebras of observables act on different legs of the tensor product of Hilbert spaces.

Similarly to the previously covered case of matrix algebras, one could go beyond the formalism of quantum mechanics and explore the possibility of a steering scenario where the description of a joint system may not be of a quantum nature while the given procedure still obeys no-signaling constraints. To realize this idea within the most general C*-algebraic framework we evoke the following definition (see also \cite{Fritz, Yan}).

\begin{definition}\label{ns}
We say that a collection of positive functionals $\Sigma=\left\{\sigma_{a|x}\right\}_{a,x}$ on a unital C*-algebra $B$ is a no-signaling assemblage if there exists a state $\sigma \in S(B)$ such that $$\sum_a\sigma_{a|x}=\sigma $$
for any $x=1,\ldots, m$. For a given scenario $(m,k)$ we denote the set of all no-signaling assemblages on $B$ by $\mathbf{nsA}(m,k,B)$.
\end{definition}

Beside evoked definitions, in the light of previous discussion motivated by the algebraic quantum field theory \cite{Haag, Haag0}, it is natural to consider yet another class of assemblages. Description of this intermediate set formalizes questions (initialized in  \cite{Nava}) concerning steering inequalities in the commuting models.

\begin{definition}\label{qc}
We say that a collection of positive functionals $\Sigma=\left\{\sigma_{a|x}\right\}_{a,x}$ on a unital C*-algebra $B$ is a quantum commuting assemblage if there exists a unital inclusion $B\subset B(H)$ (given by some $*$-homomorphism) and PVMs elements $P_{a|x}\in B'\subset B(H)$ together with $\phi\in S(B(H))$ such that
\begin{equation}\label{formula}
\sigma_{a|x}(\cdot)=\phi(P_{a|x}\cdot)
\end{equation}for all $a=1,\ldots, k$, $x=1,\ldots, m$. For a given scenario $(m,k)$ we denote the set of all quantum commuting assemblages on $B$ by $\mathbf{qcA}(m,k,B)$. 
\end{definition}

Quantum commuting assemblages may be equivalently defined by the following formula $\sigma_{a|x}(\cdot)=\phi(P_{a|x}\otimes\cdot )$ with $\phi \in S(A\otimes_{max}B)$ for some unital C*-algebra $A$. Indeed, if $\Sigma=\left\{\sigma_{a|x}\right\}_{a,x}\in \mathbf{qcA}(m,k,B)$ is represented by (\ref{formula}), then natural inclusions $\pi:B\rightarrow B(H)$ and $\tilde{\pi}:B'\rightarrow B(H)$ have commuting ranges, and as a consequence, there exists a unital completely positive extension $\tilde{\pi}:B'\otimes_{max}B\rightarrow B(H)$ such that $\tilde{\pi}(a\otimes b)=ab$ for all $a\in B', b\in B$. Therefore, $\sigma_{a|x}(\cdot)=\tilde{\phi}(P_{a|x}\otimes \cdot)$ for $\tilde{\phi}=\phi\circ \tilde{\pi}\in S(B'\otimes_{max} B)$. On the other hand, if $\sigma_{a|x}(\cdot)=\tilde{\phi}(P_{a|x}\otimes \cdot)$ for some $\tilde{\phi}\in S(A\otimes_{max}B)$ then under natural inclusion $\pi:B\rightarrow A\otimes_{max}B\subset B(H)$ with $\pi(B)=\mathds{1}\otimes B$ it is clear that $P_{a|x}\otimes \mathds{1}\in \pi(B)'$ and $\Sigma=\left\{\sigma_{a|x}\right\}_{a,x}$ admits quantum realization in commuting model. A shift from the tensor product of Hilbert spaces to commuting model is described by a change from minimal to maximal tensor products of underlying (local) C*-algebras of observables.

It is easy to observe that $$\mathbf{qA}(m,k,B)\subset \mathbf{qcA}(m,k,B)\subset \mathbf{nsA}(m,k,B)$$for a given scenario $(m,k)$ and fixed local C*-algebra $B$ of characterized (trusted) subsystem. In the remaining part of the paper, we provide a systematic analysis of the above chain of inclusions.

Note that Definitions \ref{qt} and \ref{qc} can be modified by substitution of PVMs elements $P_{a|x}$ by positive operator-valued measures (POVMs) elements $M_{a|x}$. Indeed, the collection of positive operators $M_{a|x}\in B$ such that $\sum_a M_{a|x}=\mathds{1}$ for any fixed $x$ is the most general form of a measurement representation used within quantum mechanical formalism. Before we address the possible differences between the PVM-based and POVM-based approaches we introduce a notion of the universal (unital) C*-algebra (see \cite{Blackadar2}) generated by PVMs elements $P_{a|x}$ defined by
\begin{equation}\nonumber
C^*(m,k)=\biggr \langle\left\{P_{a|x}\right\}_{a=1,x=1}^{k,m}:\sum_aP_{a|x}=\mathds{1},P_{a|x}=P^*_{a|x}=P^2_{a|x}\biggr \rangle.
\end{equation}Universal property of $C^*(m,k)$ states that for any PVMs elements $Q_{a|x}$ from some unital C*-algebra $C$, there exists a unique unital $*$-homomorphism $\pi:C^*(m,k)\rightarrow C$ such that $\pi(P_{a|x})=Q_{a|x}$. The aforementioned C*-algebra algebra can be concretely realized as a full (universal) group C*-algebra, namely
\begin{equation}\nonumber
C^*(m,k)\simeq C^*(*_{i=1}^m\mathbb{Z}_k)\simeq C^*(\mathbb{Z}_k)*\ldots *C^*(\mathbb{Z}_k)
\end{equation}where $*$ stands for free product of C*-algebras (amalgamated over identity). Note that since $C^*(\mathbb{Z}_k)\simeq \mathbb{C}^k$, $C^*(m,k)$ is $*$-isomorphic with a full free product of finite-dimensional commutative algebras. Observe then that $C^*(\mathbb{Z}_k)$ is a universal C*-algebra generated by a single PVM, i.e. by a collection of projectors $P_{a|x}$ with $a=1,\ldots, k$ such that $\sum_aP_{a|x}=\mathds{1}$. Moreover, $C^*(m,k)$ posses certain extension properties coming from the following result due to Boca \cite{Boca}.

\begin{theorem}[Boca]\label{boca}Let $A$, $B$ and $C$ be unital C*-algebras. If $\mathcal{E}_A:A\rightarrow C$ and $\mathcal{E}_B:B\rightarrow C$ are unital and completely positive maps, then there exists a unital completely positive map $\mathcal{E}:A*B\rightarrow C$ such that $\mathcal{E}|_A=\mathcal{E}_A$ and $\mathcal{E}|_B=\mathcal{E}_B$.
\end{theorem}

The previous discussion enable us to provide an equivalent characterization of quantum commuting assemblages based on the structure of universal C*-algebra $C^*(m,k)$.

\begin{theorem}\label{max}Let $\Sigma=\left\{\sigma_{a|x}\right\}_{a,x}$ be a collection of positive functionals on a unital C*-algebra $B$. The following conditions are equivalent
\begin{enumerate}
\item $\sigma_{a|x}(\cdot)=\phi(P_{a|x}\otimes \cdot)$ for some $\phi\in S(C^*(m,k)\otimes_{max}B)$ and generating projectors $P_{a|x}\in C^*(m,k)$
\item $\sigma_{a|x}(\cdot)=\phi(P_{a|x}\otimes \cdot)$ for some unital C*-algebra $A$, some $\phi\in S(A\otimes_{max}B)$ and PVMs elements $P_{a|x}\in A$
\item  $\sigma_{a|x}(\cdot)=\phi(M_{a|x}\otimes \cdot)$ for some unital C*-algebra $A$, some $\phi\in S(A\otimes_{max}B)$ and POVMs elements $M_{a|x}\in A$
\end{enumerate}
\end{theorem}
\begin{proof}Implications $(1)\Rightarrow (2)$ and $(2)\Rightarrow (3)$ are straightforward. To show that $(3)$ implies $(1)$ define a unital positive map $\mathcal{E}_x:C^*(\mathbb{Z}_k)\rightarrow A$  given by $\mathcal{E}_x(P_{a|x})=M_{a|x}$ for any choice of measurement $x=1,\ldots, m$. As $C^*(\mathbb{Z}_k)$ is commutative all of such maps  $\mathcal{E}_x$ are in fact completely positive. Since $C^*(m,k)\simeq C^*(\mathbb{Z}_k)\ast\ldots \ast C^*(\mathbb{Z}_k)$, by Theorem \ref{boca} there exists a unital completely positive map $\mathcal{E}:C^*(m,k)\rightarrow A$ such that $\mathcal{E}(P_{a|x})=M_{a|x}$ for all $a,x$. By the properties of the maximal tensor product, there exists a unital completely positive map $\tilde{\mathcal{E}}:C^*(m,k)\otimes_{max}B\rightarrow A\otimes_{max}B$ such that $\tilde{\mathcal{E}}(a\otimes b)=\mathcal{E}(a)\otimes b$ for all $a\in C^*(m,k)$ and $b\in B$. Therefore, there exists $\tilde{\phi}=\phi\circ \tilde{\mathcal{E}}\in S(C^*(m,k)\otimes_{max}B)$ such that $\sigma_{a|x}(\cdot)=\phi(P_{a|x}\otimes \cdot)$ for generating projectors $P_{a|x}\in C^*(m,k)$.
\end{proof}

Similar equivalent formulations can be obtained in the case of the quantum tensor model.

\begin{theorem}\label{min}Let $\Sigma=\left\{\sigma_{a|x}\right\}_{a,x}$ a collection of positive functionals on a unital C*-algebra $B$. The following conditions are equivalent
\begin{enumerate}
\item $\sigma_{a|x}(\cdot)=\phi(P_{a|x}\otimes \cdot)$ for some $\phi\in S(C^*(m,k)\otimes_{min}B)$ and generating projectors $P_{a|x}\in C^*(m,k)$
\item $\sigma_{a|x}(\cdot)=\phi(P_{a|x}\otimes \cdot)$ for some unital C*-algebra $A$, some $\phi\in S(A\otimes_{min}B)$ and PVMs elements $P_{a|x}\in A$
\item  $\sigma_{a|x}(\cdot)=\phi(M_{a|x}\otimes \cdot)$ for some unital C*-algebra $A$, some $\phi\in S(A\otimes_{min}B)$ and POVMs elements $M_{a|x}\in A$
\end{enumerate}
\end{theorem}
\begin{proof}
Implications $(1)\Rightarrow (2)$ and $(2)\Rightarrow (3)$ are obvious. The implication from $(3)$ to $(1)$ comes from analogous reasoning as the one given in the proof of Theorem \ref{max} together with the properties of the minimal tensor product.
\end{proof}

From now on we will mostly take advantage of the above characterizations of sets $\mathbf{qA}(m,k,B)$ and $\mathbf{qcA}(m,k,B)$ expressed with respect to the universal C*-algebra $C^*(m,k)$ and its generators $P_{a|x}$. In particular, one can now justify that both quantum commuting and quantum tensor assemblages form closed and convex sets.

\begin{proposition}\label{Hahn}Sets of no-signaling $\mathbf{nsA}(m,k,B)$, quantum commuting $\mathbf{qcA}(m,k,B)$ and quantum tensor $\mathbf{qA}(m,k,B)$ assemblages are closed (with respect to a weak-* topology) and convex.
\end{proposition}
\begin{proof}Let $B^*$ denotes a space of bounded functional on $B$ equipped with the weak-* topology. For a particular scenario with $m$ measurements and $k$ outcomes consider a dual space $\left(\bigoplus_{i=1}^{mk} B\right)^*$ of the direct sum of $mk$ copies of $B$ (once more equipped with the weak-* topology). Observe that since no-signaling assemblages are defined by linear constraints and a limit of a net of positive functionals has to be positive as well, it is clear that $\mathbf{nsA}(m,k,B)\subset \left(\bigoplus_{i=1}^{mk} B\right)^*$ is a closed convex set. Consider now a following function $F:S(C^*(m,k)\otimes_{max} B)\rightarrow \left(\bigoplus_{i=1}^{mk} B\right)^*$ given by 
\begin{equation}\nonumber
F(\phi)=(\phi(P_{1|1}\otimes\ \cdot),\phi(P_{2|1}\otimes\ \cdot),\ldots, \phi(P_{k|m}\otimes\ \cdot)).
\end{equation}Since $F(S(C^*(m,k)\otimes_{max} B))=\mathbf{qcA}(m,k,B)$ linearity and continuity of $F$ implies convexity and compactness of $\mathbf{qcA}(m,k,B)$ (as the set of states on a unital C*-algebra is compact according to Banach-Alaoglu theorem \cite{Kadison0}). The case of $\mathbf{qA}(m,k,B)$ comes from analogous reasoning with $\otimes_{min}$ instead of $\otimes_{max}$.
\end{proof}

\section{Inequivalence of steering scenarios in tensor and commuting models}

Since considerations regarding steering in a case of tensor product of Hilbert spaces and single Hilbert spaces with commuting operators provided in \cite{Nava} can be expressed in the C*-algebraic formalism of Definitions \ref{qt} and \ref{qc}, it is known that for some particular choice of a scenario $(m,k)$ and a C*-algebra $B$ quantum commuting and quantum tensor models do not coincide. In the remaining part of this section, we provide a systematic analysis of C*-algebras for which such gap can or cannot be detected. 

Note that due to Theorem \ref{max} and Theorem \ref{min}, if a unital C*-algebra $B$ of local observables is nuclear (i.e. $A\otimes _{max}B=A\otimes _{min} B$ for all C*-algebras $A$), then both quantum models of steering coincide for any scenario $(m,k)$. In particular, if $B$ is commutative (corresponds to the classical system) or $B=M_n(\mathbb{C})$ (corresponds to the finite-dimensional model of a system), then one cannot expect a gap between quantum commuting and tensor descriptions regardless of the number of measurement settings and outcomes.

On the other hand, since $\mathbb{Z}_2\ast \mathbb{Z}_2$ is amenable group, hence $C^*(2,2)\simeq C^*(\mathbb{Z}_2\ast \mathbb{Z}_2)$ is nuclear (and it is the only nuclear C*-algebra among the family  $C^*(\ast_{i=1}^m\mathbb{Z}_k)$) \cite{Brown,Pisier}.

\begin{proposition}\label{22}Let $B$ be a unital C*-algebra, then $\mathbf{qA}(2,2,B)=\mathbf{qcA}(2,2,B)$.
\end{proposition}

One could naively assume that quantum commuting and tensor models are equivalent if and only if at least one of algebras $C^*(\ast_{i=1}^m\mathbb{Z}_k)$ and $B$ is nuclear. It is, however, not the case. Indeed, note that $B(H)\otimes_{max} B(H)\neq B(H)\otimes_{min}B(H)$ with separable Hilbert space $H$  \cite{JungeP}, nevertheless based on Kirchberg result \cite{Kirchberg} we show the following proposition.

\begin{proposition}Let $B(H)$ denotes a C*-algebra of bounded operator on a (separable) Hilbert space $H$, then 
\begin{equation}\nonumber
\mathbf{qA}(m,k,B(H))=\mathbf{qcA}(m,k,B(H))
\end{equation}for all scenarios $(m,k)$.
\end{proposition}
\begin{proof}
Observe that $C^*(m,k)\simeq C^*(\mathbb{Z}_k)\ast\ldots \ast C^*(\mathbb{Z}_k)$ with $C^*(\mathbb{Z}_k)\simeq \mathbb{C}^k$ being commutative. Due to Kirchberg result \cite{Kirchberg} (see also \cite{Pisier96} and \cite{Ozawa}), if unital C*-algebras $A_i$ for $i=1,2$ satisfy $A_i\otimes_{max} B(H)=A_i\otimes_{min}B(H)$, then $\left(A_1\ast A_2\right)\otimes_{max} B(H)=\left(A_1\ast A_2\right)\otimes_{min}B(H)$. Therefore, it follows that $C^*(m,k)\otimes_{max}B(H)= C^*(m,k)\otimes_{min}B(H)$ and the proof is completed by Theorem \ref{max} and Theorem \ref{min}.
\end{proof}
This results provides an interesting physical insight. In order to observe a difference between commuting and tensor models of steering, local observer related to the characterized subsystem should not have an access to the all possible observables.

In fact one can extend the above result to cover a larger class of C*-algebras thoroughly studied in the literature. Recall that a C*-algebra has the weak expectation property (WEP)  if $$A\otimes_{max} C\subset B\otimes_{max} C$$ for any C*-algebras $C$, $B$ such that $A\subset B$ \cite{Brown}. The following proposition formulates relationship between the WEP and the universal C*-algebra $C^*(m,k)$.

 \begin{proposition}\label{WEP} Let $B$ be a C*-algebra. The following conditions are equivalent 
 \begin{enumerate}
 \item $B$ has the weak expectation property,
 \item $(C^*(m,k),B)$ is a nuclear pair for any $(m,k)\neq(2,2)$,
  \item $(C^*(m,k),B)$ is a nuclear pair for some $(m,k)\neq(2,2)$.
\end{enumerate}
\end{proposition}
\begin{proof}To justify $(1)\Rightarrow (2)$ recall that $B$ has the weak expectation property if and only if $(C^*(\mathbb{F}_\infty),B)$ is a nuclear pair \cite{Brown} - here $C^*(\mathbb{F}_\infty)$ denotes the full group C*-algebra of a free group with countably many generators. Since $C^*(m,k)$ is a free product of $C^*(\mathbb{Z}_k)\simeq\mathbb{C}^k$, it admits the lifting property (see Proposition 3.21 in \cite{Ozawa0} based on \cite{Pisier96}), i.e. for any unital C*-algebra $A$ and any closed two-sided (hence self-adjoint) ideal $I\subset A$, any unital completely positive map $\mathcal{E}:C^*(m,k)\rightarrow A/I$ lifts to the unital completely positive map $\tilde{\mathcal{E}}:C^*(m,k)\rightarrow A$ such that $\mathcal{E}=q\circ \tilde{\mathcal{E}}$ where $q:A\rightarrow A/I$ stands for the quotient map \cite{Kirchberg}. Since $C^*(m,k)\simeq C^*(\mathbb{F}_\infty)/J$ for some closed two-sided ideal $J$ there exists a $*$-isomorphism $i:C^*(m,k)\rightarrow  C^*(\mathbb{F}_\infty)/J$ and its lift $\tilde{i}:C^*(m,k)\rightarrow  C^*(\mathbb{F}_\infty)$. Consider a map $q\otimes id:C^*(\mathbb{F}_\infty)\otimes_{max}B\rightarrow C^*(m,k)\otimes_{max} B$ where $q$ is a quotient map. As $C^*(\mathbb{F}_\infty)\otimes_{max}B=C^*(\mathbb{F}_\infty)\otimes_{min}B$ one can define the following composition of unital completely positive maps $q\circ\tilde{i}\otimes id:C^*(m,k)\otimes_{min}B\rightarrow C^*(m,k)\otimes_{max}B$ such that $||x||_{max}=||q\circ\tilde{i}\otimes id(x)||_{max}\leq ||x||_{min}$ for all $x\in C^*(m,k)\otimes_{alg}B$, so $C^*(m,k)\otimes_{max}B=C^*(m,k)\otimes_{min}B$.

Implication $(2)\Rightarrow (3)$ is obvious. To show $(3)\Rightarrow (1)$ observe that due to the Ping-Pong lemma (see also Lemma D.1 in \cite{Fritz}) $\mathbb{F}_{\infty}$ is a subgroup in $*_{i=1}^m\mathbb{Z}_k$ (for $(m,k)\neq (2,2)$) and as a consequence (see Remark 7.22 in \cite{Pisier}) there exists a max-injective embedding $C^*(\mathbb{F}_{\infty})\subset C^*(*_{i=1}^m\mathbb{Z}_k)\simeq C^*(m,k)$, i.e. $C^*(\mathbb{F}_\infty)\otimes_{max}C\subset C^*(m,k)\otimes_{max}C$ for any C*-algebra $C$. Therefore, nuclearity of $(C^*(m,k),B)$ implies nuclearity of $(C^*(\mathbb{F}_\infty),B)$.
\end{proof}

The following theorem provides a link between WEP C*-algebras and steering scenarios.

\begin{theorem}\label{mat_stab}Let $B$ be a unital C*-algebra, then $C^*(m,k)\otimes_{max}B=C^*(m,k)\otimes_{min} B$ if and only if $\mathbf{qA}(m,k,M_n(B))=\mathbf{qcA}(m,k,M_n(B))$ for all $n\in \mathbb{N}$ with given $(m,k)\neq (2,2)$.
\end{theorem}
\begin{proof}If $(C^*(m,k),B)$ is a nuclear pair then quantum commuting and quantum tensor models coincide for $M_n(B)$ with any $n\in \mathbb{N}$. 
To show the converse implication consider the set  $U(B)$ of all unitaries $u \in B$ and the following subspace 
$$S=\mathrm{span}\left\{P_{a|x}\otimes u:a=1,\ldots, k, x=1,\ldots,m ,u\in U(B)\right\}\subset C^*(m,k)\otimes_{alg}B.$$According to \cite{Pisier96} (see also Proposition B.14 in \cite{Fritz}) $C^*(m,k)\otimes_{max}B=C^*(m,k)\otimes_{min}B$ if and only if norms of self-adjoint elements in $M_n(S)$ coming from embeddings $M_n(S)\subset M_n(C^*(m,k)\otimes_{max}B)$ and $M_n(S)\subset M_n(C^*(m,k)\otimes_{min}B)$ coincide for all $n\in \mathbb{N}$. Observe that any $s\in M_n(S)$ is of the form $s=\sum_{a,x,i,j}P_{a|x}\otimes u_{a,x,i,j}\otimes E_{ij}$ where $u_{a,x,i,j}\in B$ and $E_{ij}\in M_n(\mathbb{C})$ stand for matrix units. Note that since $\mathbf{qA}(m,k,M_n(B))=\mathbf{qcA}(m,k,M_n(B))$ for any state $\phi\in S(C^*(m,k)\otimes_{max} M_n(B))$ there exists $\tilde{\phi}\in  S(C^*(m,k)\otimes_{min}M_n(B))$ such that $\phi(P_{a|x}\otimes b)=\tilde{\phi}(P_{a|x}\otimes b)$ for all $b\in M_n(B)$. If so, then $\phi(s)=\tilde{\phi}(s)$ for all $s\in M_n(S)$. Since for any self-adjoint element
$$||s||_{M_n(C^*(m,k)\otimes_{max}B)}=\sup \left\{|\phi(s)|: s\in S(C^*(m,k)\otimes_{max}M_n(B)) \right\}$$and
$$||s||_{M_n(C^*(m,k)\otimes_{min}B)}=\sup \left\{|\phi(s)|: s\in S(C^*(m,k)\otimes_{min}M_n(B)) \right\}$$ hence $||s||_{M_n(C^*(m,k)\otimes_{max}B)}=||s||_{M_n(C^*(m,k)\otimes_{min}B)}$ for arbitrary self-adjoint $s\in M_n(S)$.
\end{proof}

In the light of the above theorem and Proposition \ref{WEP} one could single out unital C*-algebra with the WEP as precisely those C*-algebras for which equivalence of quantum tensor and commuting models of steering remains stable under extension of a trusted system by arbitrary finite-dimensional ancilla (modeled by the tensor product with full matrix algebra $M_n(\mathbb{C})$). These observation could be seen as one of possible physical characterization of the WEP. Observe that due to Proposition \ref{WEP} equivalence coming from Theorem \ref{mat_stab} for some fixed $(m,k)\neq (2,2)$ implies that the same equivalence holds for all $(m,k)\neq(2,2)$.

Theorem \ref{mat_stab} describes the gap between quantum tensor and commuting models of steering in a nonconstructive way. In particular, for the certain family of C*-algebras coming out of non-amenable groups. From now on while talking about discrete groups we will always assume  countability.

\begin{corollary}\label{group}Consider any pair $(m,k)\neq (2,2)$. Let $G$ be a discrete group, then $G$ is not amenable if and only if there exists $n\in \mathbb{N}$ such that $\mathbf{qA}(m,k,M_n(C^*_r(G)))\neq\mathbf{qcA}(m,k,M_n(C^*_r(G)))$ where $C^*_r(G)$ denotes the reduced group C*-algebra of $G$.
\end{corollary}
\begin{proof}If $G$ is amenable then $C^*_r(G)$ is nuclear, so $\mathbf{qA}(m,k,M_n(C^*_r(G)))=\mathbf{qcA}(m,k,M_n(C^*_r(G)))$ for any $n\in \mathbb{N}$. Conversely, if quantum tensor and quantum commuting frameworks coincide for $M_n(C_r^*(G))$ with any $n\in \mathbb{N}$, then according to Theorem \ref{mat_stab}, $(C^*(m,k),C_r^*(G))$ is a nuclear pair, so by Proposition \ref{WEP}, $C^*_r(G)$ has the weak expectation property. Due to Corollary 9.29 in \cite{Pisier} $G$ is then amenable.
\end{proof}Note that Corollary \ref{group} witnesses the existence of a gap between commuting and tensor models, in a way that is not based on particular and ad-hoc construction of assemblages (or steering inequalities) but it comes directly from the structural properties of underlying algebras. For any discrete, non-amenable group $G$ there exists $n\in \mathbb{N}$ such that $M_n(C^*_r(G))$ distinguishes quantum commuting and tensor models (for a given scenario $(m,k)$ if and only if $(m,k)\neq (2,2)$). In fact the previously known examples \cite{Nava, Yan} providing the gap between the models could be seen as a particular case (with $n=1$) of Corollary \ref{group}. Indeed, it has been shown \cite{Nava} that for $H=\ell^2(*_{i=1}^m\mathbb{Z}_2)$ there is a fixed choice of PVMs elements $Q_{b|y}\in B(H)$ such that 
\begin{equation}\label{look}
I\left(\left\{p(ab|xy)\right\}_{a,b,x,y}\right)=\frac{1}{m}\sum_{x=1}^mp(a=b|xx)-p(a\neq b|xx)\leq \frac{2\sqrt{m-1}}{m}
\end{equation}
for quantum tensor model of correlations, while $I\left(\left\{p(ab|xy)\right\}_{a,b,x,y}\right)=1$ for quantum commuting model of correlations (where $a,b=1,2$ and  $x,y=1,\ldots, m$). As these reasoning can be translated into C*-algebraic framework, it shows that $\mathbf{qA}(m,2,B)\neq \mathbf{qcA}(m,2,B)$ for C*-algebra $B=C^*_r(\ast_{i=1}^m\mathbb{Z}_2)$ as long as $m\geq 3$. Similarly, \cite{Yan} provides construction of an assemblage beyond quantum tensor description for the case of C*-probability spaces $(B,\phi)$ and assemblages constructed from faithful tracial state $\phi\in S(B)$ and $m$ freely independent self-adjoint unitaries - specifically,  for paradigmatic example $B=C_r^*(*_{i=1}^m\mathbb{Z}_2)$ with $m\geq 5$.

One can justify that $\mathbf{qA}(m,k,B)\neq \mathbf{qcA}(m,k,B)$ implies $\mathbf{qA}(m',k',B)\neq \mathbf{qcA}(m',k',B)$ for $m'\geq m$ and $k'\geq k$. Therefore, note that Corollary \ref{group} for the first time implies (non-constructively) existence of quantum commuting assemblages with no tensor realization for two-measurements scenarios beyond binary outcomes assumption, i.e. $(2,k)\neq (2,2)$.

It should be emphasized that the difference between tensor product and commuting paradigms for steering scenarios is not simply a result of difference between sets of bipartite correlations described within both approaches (such a separation follows from the refutation of Tsirelson’s conjecture \cite{Re}, however, it is not known for which values of $(m,k)$ one has $\mathcal{C}_{qa}(m,k)\neq \mathcal{C}_{qc}(m,k)$). It constitutes an independent phenomena. To see this recall that a C*-algebras $A$ is QWEP if there is a C*-algebra $B$ with the WEP and a two-sided closed ideal $I\subset B$ such that $A\simeq B/I$.
\begin{theorem}For any $(m,k)\neq (2,2)$ there exists a unital C*-algebra $B$ such that $\mathbf{qA}(m,k,B)\neq \mathbf{qcA}(m,k,B)$, but every behaviour $P=\left\{p(ab|xy)\right\}_{a,b,x,y}$ of the form 
$$p(ab|xy)=\sigma_{a|x}(M_{b|y})$$with $\left\{\sigma_{a|x}\right\}_{a,x}\in \mathbf{qcA}(m,k,B)$ and POVMs elements $M_{b|y}\in B$, fulfills $P\in \mathcal{C}_{qa}(m,k)$. In particular, if $C$ is $\mathrm{QWEP}$ but not $\mathrm{WEP}$, then there exists $n \in \mathbb{N}$ such that $B=M_n(C)$ fulfills the above.

\end{theorem}
\begin{proof}Let $B$ be an arbitrary C*-algebra with the WEP (without loss of generality we may assume that $B$ is unital) and let $I\subset B$ be an arbitrary two-sided closed ideal, such that $B/I$ is a unital C*-algebra. For any $n\in \mathbb{N}$ consider $M_n(B/I)\simeq M_n(B)/M_n(I)$. Since  $M_n(B/I)$ are QWEP, for any $\left\{\sigma_{a|x}\right\}_{a,x}\in \mathbf{qcA}(m,k,M_n(B/I))$ and POVMs elements $M_{b|y}\in M_n(B/I)$, we have $P=\left\{p(ab|xy)\right\}_{a,b,x,y}\in C_{qa}(m,k)$ for $p(ab|xy)=\sigma_{a|x}(M_{b|y})$. Indeed, there exists a UCP map $\Phi:C^*(m,k)\rightarrow M_n(B/I)$ such that $\Phi(P_{b|y})=M_{b|y}$ and since $C^*(m,k)$ admits the lifting property there exists a UCP $\tilde{\Phi}:C^*(m,k)\rightarrow M_n(B)$ such that $\Phi=q\circ \tilde{\Phi}$ where $q:M_n(B)\rightarrow M_n(B/I)$ stands for the quotient map. Therefore, $\sigma_{a|x}(M_{b|y})=\phi\circ(id\otimes q)(P_{a|x}\otimes \tilde{\Phi}(P_{b|y}))$
with $\phi\circ(id\otimes q)\in S(C^*(m,k)\otimes_{min}C^*(m,k))$ as $(C^*(m,k),M_n(B))$ is a nuclear pair since $M_n(B)$ has the WEP.

Assume for now that $\mathbf{qA}(m,k,M_n(B/I))=\mathbf{qcA}(m,k,M_n(B/I))$ for all $n\in \mathbb{N}$. Due to the steering-based characterization of WEP in Theorem \ref{mat_stab}, this would imply that $B/I$ has WEP, so in fact WEP and QWEP would be equivalent notions. However, this is not the case, as for example the reduced group C*-algebra $C^*_r(SL_2(\mathbb{Z}))$ is QWEP but not WEP \cite{Kirchberg}. 
\end{proof}

\color{black}

Finally, note that the previous discussion provides a useful tool for the analysis of tensor products of C*-algebras. If $\mathbf{qA}(m,k,M_n(B))\neq \mathbf{qcA}(m,k,M_n(B))$ for some $n$ (in particular for $n=1$) then the difference between tensor and commuting models become a witness of non-nuclearity of pairs of C*-algebras. For example, steering inequality (\ref{look}) constructed in \cite{Nava} for unital C*-algebra $B=C^*_r(*_{i=1}^m\mathbb{Z}_2)$ shows that \begin{equation}\nonumber
C^*(m,2)\otimes_{min}C_r^*(*_{i=1}^m\mathbb{Z}_2)\neq C^*(m,2)\otimes_{max}C_r^*(*_{i=1}^m\mathbb{Z}_2)
\end{equation}for any $m\geq 3$.

We conclude structural comments with a remark concerning the family of specific subalgebras that behave well under partial order imposed by the notion of positive elements. Recall that a C*-subalgebra $A\subset B$ is hereditary subalgebra of a C*-algebra $B$ if for any positive elements $a\in A$ and $b\in B$ relation $b\leq a$ implies that $b\in A$. Note that for this class of C*-algebras there exists a natural embedding $A\otimes_{max}C\subset B\otimes_{max} C$ for arbitrary C*-algebra $C$ \cite{Brown}. This property leads to the following observation based on Theorem \ref{mat_stab}.

\begin{corollary}\label{hered}Let $B$ a unital C*-algebra such that $\mathbf{qA}(m,k,M_n(B))=\mathbf{qcA}(m,k,M_n(B))$ for any $n\in \mathbb{N}$, then any unital C*-algebra $C$ that is $*$-isomorphic with some hereditary subalgebra of $B$ has WEP, hence $\mathbf{qA}(m,k,C)=\mathbf{qcA}(m,k,C)$.
\end{corollary}

Observe that for a given unital C*-algebra $B$, $M_n(B)$ can be seen a hereditary C*-subalgebra of $M_l(B)$ for $l>n$. Using Corollary \ref{hered} above characterization of WEP leads to the following result.

 \begin{corollary}\label{WEP-hered} Let $B$ be a unital C*-algebra. The following conditions are equivalent 
 \begin{enumerate}
 \item $B$ has the weak expectation property,
 \item there exists $n\in \mathbb{N}$ such that $\mathbf{qA}(m,k,M_{nl}(B))=\mathbf{qcA}(m,k,M_{nl}(B))$ for all $l\in \mathbb{N}$,
  \item there exists $n\in \mathbb{N}$ such that $\mathbf{qA}(m,k,M_{l}(B))=\mathbf{qcA}(m,k,M_{l}(B))$ for all $l\geq n$.
\end{enumerate}
\end{corollary}

We close this section with a short discussion of the relevance of Theorem \ref{mat_stab} for a large class of models considered within the framework of algebraic quantum field theory, where space-like separated subsystems are often described by hyperfinite von Neumann algebras. The following result provides a stronger version of Proposition 53 from \cite{Luijk0} (and anticipated in \cite{Scholz}). We include it here also for a reader convenience, to avoid the approximation required in the proof of the aforementioned proposition in \cite{Luijk0} which is not explicitly argued. 

\begin{corollary}Let $\mathcal{M}$ stands for a hyperfinite von Neumann algebra, then $\mathbf{qcA}(m,k,\mathcal{M})=\mathbf{qA}(m,k,\mathcal{M})$ for all $(m,k)$. As a consequence $P=\left\{p(ab|xy)\right\}_{a,b,x,y}\in \mathcal{C}_{qa}(m,k)$ for all bipartite Bell-type scenarios involving $\mathcal{M}$ as one of commuting algebras describing local subsystems.
\end{corollary}
\begin{proof}
Observe that if $\mathcal{M}$ is hyperfinite, then it is injective which shows that $\mathcal{M}$ has the weak expectation property (WEP). The statement then follows imediately from Theorem \ref{mat_stab}.
\end{proof}

Let us also remark that the provided difference between the quantum commuting and the quantum tensor models for steering remains unchanged when one restrict Definitions \ref{qt} and \ref{qc} only to normal states, i.e if $\mathbf{qA}(m,k,B)\neq \mathbf{qcA}(m,k,B)$, then respective subsets of assemblages in tensor and commuting models obtainable with measurements on states (for a joint system) defined through density matrices differs as well.
\color{black}

\section{Bipartite correlations from no-signaling steering}
Consider an experimental steering scenario described by some no-signaling assemblage $\Sigma=\left\{\sigma_{a|x}\right\}_{a,x}\in \mathbf{nsA}(m,k,B)$. If a trusted party performs additional local measurements (choosing settings at random) on subsystem $B$, then statistical characterization of such procedure would be given by some bipartite correlations admitting no-signaling constrains. Indeed, for POVMs elements $M_{b|y}\in B$ one obtains behaviour defined by $p(ab|xy)=\sigma_{a|x}(M_{b|y})$ and it is easy to see that 
\begin{equation}\nonumber
\sum_ap(ab|xy)=p(b|y),\ \ \ \sum_bp(ab|xy)=p(a|x)
\end{equation}for arbitrary measurements and outcomes, i.e. the marginal conditional probabilities for both subsystems are well defined.
The natural question arises whether such $P=\left\{p(ab|xy)\right\}_{a,b,x,y}\in \mathcal{C}_{ns}(m,k)$ can be post-quantum (i.e. whether it does not admit quantum realization in commuting model).

It is a well-know fact that bipartite no-signaling correlations form a polytope that even in the simplest scenario (with two binary measurements, i.e. $m=2,k=2$) does not coincide with a closed convex sets of all correlations admitting quantum commuting model \cite{PR}. One can say even more - there is no quantum commuting realization of extreme points of the no-signaling polytope unless the case where extreme point admits a local hidden variable (LHV) model \cite{RTHHPRL}. For that reason one could in principle expect the separation between correlations obtained by measurements performed on quantum and no-signaling assemblages.

Nevertheless, the following theorem provides a no-go type result regarding the nature of correlations that can be realized within no-signaling steering in the C*-algebraic framework.

\begin{theorem}\label{main}Let $\Sigma=\left\{\sigma_{a|x}\right\}_{a,x}\in \mathbf{nsA}(m,k,B)$ be any no-signaling assemblage and $N_{b|y}\in B$ denote elements of any POVMs. Consider behaviour $P=\left\{\sigma_{a|x}(N_{b|y})\right\}_{a,b,x,y}$, then $P\in \mathcal{C}_{qc}(m,k)$.
\end{theorem}
\begin{proof}Consider a particular element $\sigma_{a|x}$ from a given no-signaling assemblage $\Sigma=\left\{\sigma_{a|x}\right\}_{a,x}$ and define the following state $\sigma=\sum_a \sigma_{a|x}$. Clearly, $\sigma_{a|x}\leq \sigma$. Define left ideals $I_\sigma=\left\{b\in B:\sigma(b^*b)=0\right\}$ and $I_{\sigma_{a|x}}=\left\{b\in B:\sigma_{a|x}(b^*b)=0\right\}$. Consider now a GNS   representation $(\pi, H,\Omega)$ of a unital C*-algebra $B$ defined with respect to $\sigma$. It easy to see that $I_{\sigma}\subset I_{\sigma_{a|x}}$ and as a consequence the following formula
$$\omega_{a|x}(\pi(b)\Omega,\pi(c)\Omega)=\sigma_{a|x}(b^*c)$$
give a rise to the well-defined sesquilinear form $\omega_{a|x}$ on a dense subspace $\left\{\pi(b)\Omega:b\in B\right\}\subset H$.

Next part of our reasoning is based on the version of Nikodym-Sakai theorem (see \cite{Kadison}). Indeed, the following chain of inequalities$$|\omega_{a|x}(\pi(b)\Omega,\pi(c)\Omega)|^2=|\sigma_{a|x}(b^*c)|^2\leq \sigma_{a|x}(b^*b)\sigma_{a|x}(c^*c)\leq \sigma(b^*b)\sigma(c^*c)=||\pi(b)\Omega||^2||\pi(c)\Omega||^2$$
valid for all $b,c\in B$ shows that $\omega_{a|x}$ is a bounded sesquilinear form on $\left\{\pi(b)\Omega:b\in B\right\}$ so it can be uniquely extend to the bounded sesquiliner form $\tilde{\omega}_{a|x}$ on $H$. Hence there exists a positive operator $M_{a|x}\in B(H)$ such that 
\begin{equation}\nonumber
\tilde{\omega}_{a|x}(\xi,\eta)=\langle\xi,M_{a|x}\eta\rangle.
\end{equation}Note that in particular,
\begin{equation}\nonumber
\langle\pi(c)\Omega,M_{a|x}\pi(d)^*\pi(b)\Omega\rangle=\sigma_{a|x}(c^*d^*b)=\sigma_{a|x}((dc)^*b)=\langle\pi(d)\pi(c)\Omega,M_{a|x}\pi(b)\Omega\rangle
\end{equation}for any $b,c,d\in B$, so $M_{a|x}\in \pi(B)'$ since $\left\{\pi(b)\Omega:b\in B\right\}$ is dense in $H$. Therefore, $\sigma_{a|x}(b)=\langle\Omega,\pi(b)M_{a|x}\Omega\rangle$ for any $b\in B$.

Now, observe that for any $b,c\in B$ we get
\begin{equation}\nonumber
\langle\pi(b)\Omega,\pi(c)\Omega\rangle=\sigma(b^*c)=\sum_a\sigma_{a|x}(b^*c)=\sum_a\langle\pi(b)\Omega,M_{a|x}\pi(c)\Omega\rangle=\Biggl\langle\pi(b)\Omega,\left(\sum_aM_{a|x}\right)\pi(c)\Omega\Biggl\rangle
\end{equation}and as a consequence 
\begin{equation}\nonumber
\Biggl\langle\pi(b)\Omega,\left(\mathds{1}-\sum_aM_{a|x}\right)\pi(c)\Omega\Biggl\rangle=0
\end{equation}for all $b,c\in B$. As $\left\{\pi(b)\Omega:b\in B\right\}$ is dense in $H$, it follows that $\sum_a M_{a|x}=\mathds{1}$, so $M_{a|x}$ form a POVM for any $x$. If so, then for any POVMs elements $N_{b|y}\in B$ we obtain
\begin{equation}\nonumber
p(ab|xy)=\sigma_{a|x}(N_{b|y})=\langle\Omega,\pi(N_{b|y})M_{a|x}\Omega\rangle.
\end{equation}where $\sum_b \pi(N_{b|y})=\mathds{1}$ (as $\pi$ is a unital $*$-homomorphism) and the proof is concluded.
\end{proof}

In fact analysis of the above proof leads to the stronger statement.

\begin{theorem}\label{main_2_two}
Let $B$ be a unital C*-algebra, then $\mathbf{qcA}(m,k,B)=\mathbf{nsA}(m,k,B)$ for any $m,k$.
\end{theorem}
\begin{proof}As $\mathbf{qcA}(m,k,B)\subset \mathbf{nsA}(m,k,B)$ it is enough to show the converse inclusion. Consider $\Sigma=\left\{\sigma_{a|x}\right\}_{a,x}\in \mathbf{nsA}(m,k,B)$, then according to proof of Theorem \ref{main}, there exists a unital $*$-homomorphism $\pi:B\rightarrow B(H)$, a vector state $\phi\in S(B(H))$ and POVMs elements $M_{a|x}\in \pi(B)'$ such that $\sigma_{a|x}(\cdot)=\phi(M_{a|x}\pi(\cdot))$. Let $\tilde{\pi}:C^*(m,k)\rightarrow \pi(B)'\subset B(H)$ denotes a unital completely positive map such that $\tilde{\pi}(P_{a|x})=M_{a|x}$ for all generating projectors $P_{a|x}$ (existence of this map follows from Theorem \ref{boca}). Since $\pi$ and $\tilde{\pi}$ have commuting ranges there exists a unital completely positive map $\Pi:C^*(m,k)\otimes_{max}B\rightarrow B(H)$ such that $\Pi(P_{a|x}\otimes b)=M_{a|x}\pi(b)$ for any $b\in B$. Therefore, there exists a state $\tilde{\phi}=\phi\circ \Pi
\in S(C^*(m,k)\otimes_{max}B)$ such that $\sigma_{a|x}(\cdot)=\tilde{\phi}(P_{a|x}\otimes\cdot )$.
\end{proof}

The equivalence of quantum commuting and no-signaling constraints for assemblages immediately leads to the following characterization of the set of bipartite correlations that admit a quantum commuting realization.
\begin{corollary}\label{correlations_col}
Behaviour $P$ admits a realization in a quantum commuting model, i.e. $P=\left\{p(ab|xy)\right\}_{a,b,x,y}\in \mathcal{C}_{qc}(m,k)$ if and only if there exists a unital C*-algebra $B$ together with POVMs elements $M_{b|y}\in B$ and a no-signaling assemblage $\Sigma=\left\{\sigma_{a|x}\right\}_{a,x}\in \mathbf{nsA}(m,k,B)$ such that
\begin{equation}\nonumber
p(ab|xy)=\sigma_{a|x}(M_{b|y}).
\end{equation}
\end{corollary}

The above discussion provides a generalization of Theorem \ref{Gisin} and serves as a closure to no-go results concerning the existence of post-quantum (yet no-signaling) steering in the bipartite setting.
For the most general description of a local trusted subsystem allowed by quantum mechanics,
no-signaling constraints already imply the quantum nature of assemblages obtained in a steering
procedure.

We conclude this section with remarks on applications of the preceding description (expressed in terms of no-signaling conditions) of quantum commuting assemblages and correlations. Notably, Corollary \ref{correlations_col} resolves an open problem regarding the image of the set of quantum commuting correlations $\mathcal{C}_{qc}$ under a specific mapping between nonlocal scenarios and prepare-and-measure scenarios, as considered in \cite{Wright}. Furthermore, following the public release of the preliminary version of this manuscript, the same characterization as presented in Theorem \ref{main} has been employed as a tool in establishing an upper bound on the quantum values of so-called compiled nonlocal games \cite{Kulpe}.

\section{Steering inequalities}

We will show that steering inequalities like the one given by formula (\ref{look}) are generic tools enabling us to separate sets of quantum commuting and quantum tensor assemblages in cases where both models do not coincide.

Assume that $\mathbf{qA}(m,k,B)\neq\mathbf{qcA}(m,k,B)$ and choose $\left\{\tilde{\sigma}_{a|x}\right\}_{a,x}\in \mathbf{qcA}(m,k,B)$ such that  $\left\{\tilde{\sigma}_{a|x}\right\}_{a,x}\notin \mathbf{qA}(m,k,B)$. Due to Proposition \ref{Hahn}, according to the Hahn-Banach separation theorem (see for example \cite{Kadison0}) and Proposition 4.44 in \cite{Nerven} there exist a continuous (with respect to weak-* topology) linear functional $f:\left(\bigoplus_{i=1}^{mk} B\right)^*\rightarrow \mathbb{C}$ and a constant $\beta \in \mathbb{R}$ such that
\begin{equation}
\forall_{\left\{\sigma_{a|x}\right\}_{a,x}\in \mathbf{qA}(m,k,B)}\ \  \Re \left(f\left(\left\{\sigma_{a|x}\right\}_{a,x}\right)\right)\leq \beta <  \Re \left(f\left(\left\{\tilde{\sigma}_{a|x}\right\}_{a,x}\right)\right),
\end{equation}where $\Re$ stands for a real part of a complex number. Note (see Proposition 4.43 in \cite{Nerven}) that for arbitrary functional $f$ of this form there exists an element $\left\{f_{a|x}\right\}_{a,x}\in \bigoplus_{i=1}^{mk} B$ such that 
\begin{equation}
f\left(\left\{\varphi_{a|x}\right\}_{a,x}\right)=\sum_{a,x}\varphi_{a|x}(f_{a|x})
\end{equation}for any $\left\{\varphi_{a|x}\right\}_{a,x}\in \left(\bigoplus_{i=1}^{mk} B\right)^*$. Because of that we obtain 
\begin{equation}\label{inequality_thm}
\forall_{\left\{\sigma_{a|x}\right\}_{a,x}\in \mathbf{qA}(m,k,B)}\ \  \sum_{a,x}\sigma_{a|x}(h_{a|x}) \leq \beta<\sum_{a,x}\tilde{\sigma}_{a|x}(h_{a|x})
\end{equation}with hermitian elements $h_{a|x}=\frac{f_{a|x}+f^*_{a|x}}{2}$ for all $x=1,\ldots, m, a=1,\ldots, k$. The following theorem summarizes the above discussion.
\begin{theorem}\label{ineQQ}Let $\tilde{\Sigma}=\left\{\tilde{\sigma}_{a|x}\right\}_{a,x}\in \mathbf{qcA}(m,k,B)$, then $\tilde{\Sigma}\notin \mathbf{qA}(m,k,B)$ if and only if there exist a constant $\beta \in \mathbb{R}$ and hermitian elements $h_{a|x}\in B$ such that (\ref{inequality_thm}) holds.
\end{theorem}

The notion of steering inequality leads to sufficient condition for equivalence of quantum commuting and tensor descriptions.

\begin{proposition}\label{aux}Let $B$ be a unital C*-algebra. Consider arbitrary POVMs elements $N_{b|y}\in B$ with $b=1,\ldots, k+1$ and $y=1,\ldots, m$. Denote by $C^*(\mathcal{N})\subset B$ a C*-algebra generated by $\mathcal{N}=\left\{N_{b|y}\right\}_{b,y}$. If $\mathbf{qA}(m,k,C^*(\mathcal{N}))=\mathbf{qcA}(m,k,C^*(\mathcal{N}))$ for all choices of such POVMs, then $\mathbf{qA}(m,k,B)=\mathbf{qcA}(m,k,B)$.
\end{proposition}
\begin{proof}
Let $\mathbf{qA}(m,k,B)\neq\mathbf{qcA}(m,k,B)$, then according to the above discussion there exist $\left\{\tilde{\sigma}_{a|x}\right\}_{a,x}\in \mathbf{qcA}(m,k,B)$ and hermitian elements
$h_{a|x} \in B$ such that
\begin{equation}\nonumber
\sum_{a,x}\sigma_{a|x}(h_{a|x}) <\sum_{a,x}\tilde{\sigma}_{a|x}(h_{a|x})
\end{equation}for all $\left\{\sigma_{a|x}\right\}_{a,x}\in \mathbf{qA}(m,k,B)$. One can choose constants $\alpha, \beta>0$ in such a way that $\beta(h_{a|x}+\alpha\mathds{1})\geq 0 $ and $\sum_{a=1}^k \beta(h_{a|x}+\alpha\mathds{1})\leq \mathds{1}$ for all $a=1,\ldots, k,x=1,\ldots, m$. Define POVMs elements in $B$ as $N_{b|y}=\beta(h_{b|y}+\alpha\mathds{1})$ for all $b=1,\ldots, k,y=1,\ldots, m$ while $N_{k+1|y}=\mathds{1}-\sum_{b=1}^kN_{b|y}$ for all $y=1,\ldots, m$, and consider $C^*(\mathcal{N})$ with $\mathcal{N}=\left\{N_{b|y}\right\}_{b,y}$. Note as well that by the above construction
\begin{equation}\label{inequality_new}
\sum_{a=1}^k\sum_{x=1}^m\sigma_{a|x}(N_{a|x}) <\sum_{a=1}^k\sum_{x=1}^m\tilde{\sigma}_{a|x}(N_{a|x})
\end{equation}for all $\left\{\sigma_{a|x}\right\}_{a,x}\in \mathbf{qA}(m,k,B)$. Observe that by restricting functionals on $B$ to $C^*(\mathcal{N})$ one can see $\left\{\tilde{\sigma}_{a|x}\right\}_{a,x}$ as an element in $\mathbf{qcA}(m,k,C^*(\mathcal{N}))$. Assume now that $\mathbf{qA}(m,k,C^*(\mathcal{N}))=\mathbf{qcA}(m,k,C^*(\mathcal{N}))$. If so, there exists $\left\{\hat{\sigma}_{a|x}\right\}_{a,x}\in \mathbf{qA}(m,k,B)$ (by extension of states on $A\otimes_{min}C^*(\mathcal{N})$ to states on $A\otimes_{min}B$) such that $\hat{\sigma}_{a|x}(N_{b|y})=\tilde{\sigma}_{a|x}(N_{b|y})$ which contradicts  (\ref{inequality_new}), hence  $\mathbf{qA}(m,k,C^*(\mathcal{N}))\neq \mathbf{qcA}(m,k,C^*(\mathcal{N}))$. Contraposition gives the desired statement.

\end{proof}
Observe that the sufficient condition of Proposition \ref{aux} requires comparison of sets of assemblages for separable C*-algebras regardless of separability of $B$.

We will conclude this section by discussing relation between quantum steering and the notion of nonlocal games. Consider a scenario when the referee asks questions to two separated parties $A$ and $B$. The referee can choose questions $x=1,\ldots, m_A$ and $y=1,\ldots, m_B$ for both parties respectively, according to the given probability of distribution $\pi(x,y)$. Each party returns an answer respectively choosing  $a=1,\ldots,k_A$ and $b=1,\ldots,k_B$. The goal of two parties is to maximize the success probability of answering according to the provided verification function $V:\left\{1,\ldots,k_A\right\}\times \left\{1,\ldots,k_B\right\}\times \left\{1,\ldots,m_A\right\}\times\left\{1,\ldots,m_B\right\}\rightarrow \left\{0,1\right\}$. Any such bipartite nonlocal game $\mathcal{G}$ can be therefore identified with a tuple $(m_A,m_B,k_A,k_B,\pi, V)$. We assume that the parties do not communicate, but they can agree on a certain strategy of answering beforehand (as the rules of the game are known to the parties). 

If both parties share a joint quantum system, one can consider a strategy relating questions and answers to the settings and outcomes of local measurements. In that case a given strategy can be described as a behaviour $P=\left\{p(ab|xy)\right\}_{a,b,x,y}$ obtainable under allowed model of a bipartite quantum system. Optimizing over all strategies permitted within the quantum tensor and the quantum commuting models we obtain the well-known notions of quantum and quantum commuting values of nonlocal games $\omega_t(\mathcal{G})$ given as

$$\omega_{t}(\mathcal{G})=\sup_{P\in \mathcal{C}_t(m_A,m_B,k_A,k_B)}\sum_{x=1}^{m_A}\sum_{y=1}^{m_B}\sum_{a=1}^{k_A}\sum_{b=1}^{k_B}\pi(x,y)V(a,b,x,y)p(ab|xy)$$for $t=q,qc$.

In the case of steering scenarios one can modify the above premise and ask about a value of a nonlocal game in a restricted case, when the second subsystem is fully characterized (described by a given unital C*-algebra $B$) and a related party can only perform a fixed set of measurements $\mathcal{N}=\left\{N_{b|y}\right\}_{b,y}\subset B$ to obtain a bipartite behavior $P=\left\{p(ab|xy)=\sigma_{a|x}(N_{b|y})\right\}_{a,b,x,y}$ defining a chosen strategy. This idea leads to the following notion of a restricted value of nonlocal game $\omega_{t,\mathcal{N}}(\mathcal{G})$ obtainable through optimization over possible assemblages $\Sigma=\left\{\sigma_{a|x}\right\}_{a,x}$ permitted within one of different quantum descriptions of bipartite systems
$$\omega_{t,\mathcal{N}}(\mathcal{G})=\sup_{\Sigma\in \mathbf{tA}(m_A,k_A,B)}\sum_{x=1}^{m_A}\sum_{y=1}^{m_B}\sum_{a=1}^{k_A}\sum_{b=1}^{k_B}\pi(x,y)V(a,b,x,y)\sigma_{a|x}(N_{b|y})$$where $t=q,qc$. With that we can rephrase results concerning steering inequalities within a different language. 
\begin{proposition} 
 $\mathbf{qcA}(m,k,B)\neq \mathbf{qA}(m,k,B)$ if and only if there exists a nonlocal game $\mathcal{G}$ and choice of POVMs elements $N_{b|y}\in B$ such that $\omega_{q,\mathcal{N}}(\mathcal{G})<\omega_{qc,\mathcal{N}}(\mathcal{G})$.
\end{proposition}
\begin{proof}
As one implication is trivial, it is enough to consider the case in which 
$\mathbf{qcA}(m,k,B)\neq \mathbf{qA}(m,k,B)$. If so, then there exists $\tilde{\Sigma}=\left\{\tilde{\sigma}_{a|x}\right\}_{a,x}\in \mathbf{qcA}(m,k,B)$ fulfilling (\ref{inequality_new}). Consider a nonlocal game $\mathcal{G}=(m,m,k,k+1,\pi,V)$ defined by $\pi(x,y)=\frac{1}{m}\delta_{xy}$ and $V(a,b,x,y)$ given as
$$
V(a,b,x,y)=
\begin{cases}
 1 & \text{if }  a=b\\
0  & \text{otherwise} 
\end{cases}.
$$Observe that for $t=q,qc$ and $\mathcal{N}=\left\{N_{b|y}\right\}_{b,y}$ we have 
\begin{align}
\omega_{t,\mathcal{N}}(\mathcal{G})&=\sup_{\Sigma\in \mathbf{tA}(m,k,B)}\sum_{x,y=1}^m\sum_{a=1}^{k}\sum_{b=1}^{k+1}\pi(x,y)V(a,b,x,y)\sigma_{a|x}(N_{b|y})\nonumber \\
&=\sup_{\Sigma\in \mathbf{tA}(m,k,B)}\sum_{x=1}^m\sum_{a=1}^{k}\sigma_{a|x}\left(\sum_{y=1}^m\sum_{b=1}^{k+1}\pi(x,y)V(a,b,x,y)N_{b|y}\right) \nonumber\\
&=\sup_{\Sigma\in \mathbf{tA}(m,k,B)}\frac{1}{m}\sum_{x=1}^m\sum_{a=1}^{k}\sigma_{a|x}(N_{a|x})\nonumber
\end{align}and the statement follows from (\ref{inequality_new}).
\end{proof}
Note that the nonlocal game $\mathcal{G}$ constructed in the previous proof admits a perfect deterministic strategy (i.e. $\omega_{loc}(\mathcal{G})=1$), however, under additional restriction connecting answers of the second party to outcomes of certain fixed measurements, it becomes nontrivial.

\section{Permanence properties}

In what follows we will show that the family of unital C*-algebras for which there is no difference between the quantum commuting model and the quantum tensor model of steering posses certain permanence properties. 

\begin{proposition}Let $\left\{B_i\right\}_{i=1}^n$ be a family of unital C*-algebras, then $\mathbf{qA}(m,k,B_i)=\mathbf{qcA}(m,k,B_i)$ for all $i=1,\ldots, n$ if and only if $\mathbf{qA}(m,k,\oplus_{i=1}^nB_i)=\mathbf{qcA}(m,k,\oplus_{i=1}^nB_i)$.
\end{proposition}
\begin{proof}
The statement follows from the fact that any $\left\{\sigma_{a|x}\right\}_{a,x}\in \mathbf{qcA}(m,k,\oplus_{i=1}^nB_i)$ is given by $\sigma_{a|x}=\oplus_{i=1}^n\lambda_i\sigma^{(i)}_{a|x}$ with $\left\{\sigma^{(i)}_{a|x}\right\}_{a,x}\in \mathbf{qcA}(m,k,B_i)$ and $\sum_{i=1}^n \lambda_i=1$, $\lambda_i\geq 0$ for $i=1,\ldots, n$, and the following property of the minimal tensor product
$$C^*(m,k)\otimes_{min}\left(\oplus_{i=1}^nB_i\right) \simeq \oplus_{i=1}^n C^*(m,k)\otimes_{min} B_i.$$
\end{proof}

\begin{proposition}Let $B$ be a unital C*-algebras with a directed family of unital C*-subalgebras $\left\{B_i\right\}_{i\in I}$ such that $\overline{\cup_{i\in I}B_i}=B$. If $\mathbf{qA}(m,k,B_i)=\mathbf{qcA}(m,k,B_i)$ for all $i \in I$, then $\mathbf{qA}(m,k,B)=\mathbf{qcA}(m,k,B)$.
\end{proposition}
\begin{proof}Assume that $\left\{\tilde{\sigma}_{a|x}\right\}_{a,x}\in \mathbf{qcA}(m,k,B)$ and $\mathbf{qA}(m,k,B)\neq \mathbf{qcA}(m,k,B)$ then according to Theorem \ref{ineQQ} we may consider the steering inequality (\ref{inequality_thm}). Let $\tilde{\beta}=\sum_{a,x}\tilde{\sigma}_{a|x}(h_{a|x})$ and put $0<\epsilon <\frac{\tilde{\beta}-\beta}{2mk}$. There exist self-adjoint elements $h_{a|x}^{(i)}\in B_i$ for some $i\in I$ such that $\norm{h_{a|x}^{(i)}-h_{a|x}}<\epsilon$ for all $a=1,\ldots, k$, $x=1,\ldots,m$. If so, then by the triangle inequality 
\begin{equation}
\forall_{\left\{\sigma_{a|x}\right\}_{a,x}\in \mathbf{qA}(m,k,B_i)}\ \  \sum_{a,x}\sigma_{a|x}(h_{a|x}^{(i)}) < \frac{\beta+\tilde{\beta}}{2}<\sum_{a,x}\tilde{\sigma}_{a|x}(h_{a|x}^{(i)}).
\end{equation}This would imply that $\mathbf{qA}(m,k,B_i)\neq \mathbf{qcA}(m,k,B_i)$ leading to a contradiction.
\end{proof}

\color{black}
\section{No-go results for free probability}

The above discussion leading to equivalence between no-signaling and quantum commuting framework implies, in the light of \cite{Yan}, certain no-go results concerning C*-probability spaces and free independence (see \cite{Nica} for comprehensive introduction to this topic). For the reader convenience we will briefly recall construction of an assemblage discussed in Theorem 2.2 from \cite{Yan}.

Consider a unital C*-algebra $B$ together with a tracial state $\phi$ and a family $u_1,\ldots, u_m\in B$ of self-adjoint unitaries such that $m\geq 5$. Define the following no-signaling assemblage $\Sigma=\left\{\sigma_{a|x}\right\}_{a,x}\in \mathbf{nsA}(m,2,B)$ given by
\begin{equation}\nonumber
\sigma_{a|x}(\cdot)=\phi\left(\frac{\mathds{1}+(-1)^{(a+1)}u_x}{2}\cdot\right).
\end{equation}Note that 
\begin{equation}\label{YanYin}
\sum_{x=1}^m\sigma_{x}(u_x)=m
\end{equation}when $\sigma_x=\sigma_{1|x}-\sigma_{2|x}$. Consider now a quantum commuting realization $\sigma_{a|x}(\cdot)=\varphi(P_{a|x}\otimes \cdot)$ with some $\varphi\in S(C^*(m,2)\otimes_{max}B)$. Assume additionally that unitaries $\tilde{u}_x=(P_{1|x}-P_{2|x})\otimes u_x$ are freely independent with respect to (possibly some other) faithful state $\tilde{\varphi}$ while $\tilde{\varphi}(\tilde{u}_x)=0$ for $x=1,\ldots, m$. Due to inequality (C4) in \cite{Yin} (see also \cite{Collins}) assumed free independence leads to the following upper bound
\begin{equation}\nonumber
\normm{\sum_{x=1}^m\sigma_{x}(u_x)}=\normm{\varphi\Big(\sum_{x=1}^m(P_{1|x}-P_{2|x})\otimes u_x\Big)}\leq \norm{\sum_{x=1}^m\tilde{u}_x }_{max}\leq 2\sqrt{m}
\end{equation}which contradicts formula (\ref{YanYin}). Since all no-signaling assemblages admit quantum commuting realization, this implies that the additional assumption cannot be fulfilled.

Moreover, if unitaries $u_1,\ldots, u_m$ are freely independent with respect to faithful state $\phi$ with $\phi(u_x)=0$, then one can show that unitaries $\tilde{u}_x=(P_{1|x}-P_{2|x})\otimes u_x\in C^*(m,2)\otimes_{min}B$ for $x=1,\ldots, m$ are freely independent with respect to the faithful product state $\tilde{\phi}\otimes \phi$ on $C^*(m,2)\otimes_{min}B$ while $\tilde{\phi}\otimes\phi(\tilde{u}_x)=0$ when $\tilde{\phi}$ stands for some fixed faithful state on $ C^*(m,2)$ (such a state exists since $C^*(m,2)$ is separable) - this once more leads to the contradiction, when quantum commuting and tensor models coincide for a given $B$ and $\phi$ is tracial.

Those observations can be summarized in the form of the following corollaries of Theorem \ref{main_2_two}.

\begin{corollary}Let $(B,\phi)$ be a C*-probability space with $\phi$ being a tracial state and $u_1,\ldots, u_m\in B$ being self-adjoint unitaries with $m\geq 5$. There is no faithful state $\varphi$ on $C^*(m,2)\otimes_{max}B$ such that unitaries $\tilde{u}_x=(P_{1|x}-P_{2|x})\otimes u_x$ for $x=1,\ldots, m$ form a family of freely independent elements with respect to $\varphi$ while $\varphi(\tilde{u}_x)=0$ for all $x=1,\ldots, m$.
\end{corollary}

\begin{corollary} Let $B$ be a unital C*-algebra such that $\mathbf{qA}(m,2,B)=\mathbf{qcA}(m,2,B)$ for some $m\geq 5$. There is no faithful tracial state $\phi$ on $B$ such that there exists a family $u_1,\ldots, u_m\in B$ of self-adjoint unitaries with $\phi(u_x)=0$ for all $x=1,\ldots, m$ that are freely independent with respect to $\phi$.
\end{corollary}

Note that in particular the above corollary applies to unital C*-algebras $B$ with the WEP. As a unital C*-algebra $B$ admits a faithful tracial state if and only if it embeds into a type $\mathrm{II}_1$ factor, the above result can be restated in the following way.

\begin{corollary} Let $\mathcal{M}$ stands for a type $\mathrm{II}_1$  factor with a faithful tracial state $\tau$. If $B\subset \mathcal{M}$ is a unital C*-subalgebra such that $\mathbf{qA}(m,2,B)=\mathbf{qcA}(m,2,B)$ for some $m\geq 5$, then there is no family $u_1,\ldots, u_m\in B$ of self-adjoint unitaries with $\tau(u_x)=0$ for all $x=1,\ldots, m$ that are freely independent with respect to $\tau$.
\end{corollary}

\section{Operator system framework}

Note that obtained characterization of classes of assemblages can be equivalently stated in the language of operator systems (see \cite{Paulsen,Kavruk2}), similarly to the case of bipartie correlations \cite{Fritz,Junge, Harris2}. Indeed, define the following operator system 
\begin{equation}\nonumber
\mathcal{F}_{m,k}=\mathrm{span} \left\{P_{a|x}:a=1,\ldots, k, x=1,\ldots, m\right\} \subset C^*(m,k).
\end{equation}We will show that $\mathcal{F}_{m,k}$ plays an analogous role to $C^*(m,k)$ in the C*-algebraic description of steering.

In what follows, for any pair of operator systems $\mathcal{S},\mathcal{T}$, we will consider standard notions of minimal, commuting and maximal tensor products of operator systems denoted respectively by $\mathcal{S}\otimes_{MIN}\mathcal{T}$, $\mathcal{S}\otimes_{C}\mathcal{T}$ and $\mathcal{S}\otimes_{MAX}\mathcal{T}$ (see \cite{Kavruk2} for extensive introduction concerning tensor products of operator systems and their properties). Moreover, let $S(\mathcal{S})$ denotes a set of states on operator system $\mathcal{S}$, i.e. a collection of linear maps $\phi:\mathcal{S}\rightarrow \mathbb{C}$ such that $\phi(\mathcal{S}_ +)\subset \mathbb{R}_+$ and $\phi(e)=1$ where $\mathcal{S}_+\subset \mathcal{S}$ stands for the cone of positive elements and $e$ is the Archimedean matrix order unit related to operator system $\mathcal{S}$ \cite{Kavruk2}. One can show the following theorem.

\begin{theorem}Let $B$ stands for an arbitrary unital C*-algebra. Assemblage $\Sigma=\left\{\sigma_{a|x}\right\}_{a,x}$ on $B$ admits a quantum tensor model if and only if $\sigma_{a|x}(\cdot)=\phi(P_{a|x}\otimes \cdot)$ for some $\phi\in S(\mathcal{F}_{m,k}\otimes_{MIN}B)$ and it admits a quantum commuting model if and only if $\sigma_{a|x}(\cdot)=\phi(P_{a|x}\otimes \cdot)$ for some $\phi\in S(\mathcal{F}_{m,k}\otimes_{MAX}B)$. 
\end{theorem}
\begin{proof}
According to Theorem 4.4 in \cite{Kavruk2}, one can write $\mathcal{F}_{m,k}\otimes_{MIN}B\subset C^*(m,k)\otimes_{min} B$ on the level of operator system structure. Because of that, any state $\phi \in S(C^*(m,k)\otimes_{min} B)$ can be restricted to the state $\phi \in S(\mathcal{F}_{m,k}\otimes_{MIN}B)$. Conversely, due to Arveson theorem (see Theorem 7.5 in \cite{Paulsen}), any state $\phi \in S(\mathcal{F}_{m,k}\otimes_{MIN}B)$ can be extended to the state on the minimal tensor products of C*-algebras $C^*(m,k)\otimes_{min} B$, hence the first equivalence is proved.

To show the second one, observe that 
\begin{equation}\nonumber
 S(C^*(m,k)\otimes_{max}B)   \subset  S(\mathcal{F}_{m,k}\otimes_{C}B) \subset  S(\mathcal{F}_{m,k}\otimes_{MAX}B)
\end{equation}where we identify states on $C^*(m,k)\otimes_{max}B$ with their restrictions to $C^*(m,k)\otimes_{alg}B$. On the other hand, if $\sigma_{a|x}(\cdot)=\phi(P_{a|x}\otimes \cdot)$ for some $\phi \in S(\mathcal{F}_{m,k}\otimes_{MAX}B)$, then clearly $\Sigma=\left\{\sigma_{a|x}\right\}_{a,x}\in \mathbf{nsA}(m,k,B)$. Since $\mathbf{nsA}(m,k,B)=\mathbf{qcA}(m,k,B)$ by Theorem \ref{main_2_two}, we obtain the second equivalence. 
\end{proof}

The above theorem leads to the immediate corollary.

\begin{corollary}\label{92}For any unital C*-algebra $B$ and any $m,k$, $\mathbf{qA}(m,k,B)=\mathbf{qcA}(m,k,B)$ if and only if $S(\mathcal{F}_{m,k}\otimes_{MIN}B)=S(\mathcal{F}_{m,k}\otimes_{MAX}B)$.

\end{corollary}

Observe that since $\mathcal{F}_{m,k}\otimes_{MAX}B=\mathcal{F}_{m,k}\otimes_{C}B$ for any unital C*-algebra $B$ (see Theorem 6.7 in \cite{Kavruk2}) in fact one may not distinguish between maximal and commuting tensor products of operator systems in the previous discussion.

Applying Corollary \ref{92} together with Theorem \ref{mat_stab} one can in particular reconstruct equivalence between conditions $(1)$ and $(3)$ of Theorem $4.3$ in \cite{Farenick} and obtain an alternative condition for the WEP.
\begin{corollary}Let $(m,k)\neq (2,2)$. A unital C*-algebra $B$ has the WEP if and only $\mathcal{F}_{m,k}\otimes_{MIN}B=\mathcal{F}_{m,k}\otimes_{MAX}B$.
\end{corollary}
\color{black}

\section{Connection to Tsirelson's problem}

Discussed scenarios of steering in a C*-algebraic framework not only are inspired by the original Tsirelson's problem for correlations, but they may provide a direct attempt at disproving this conjecture as well (without a need for computational framework of \cite{Re}).
To see this, recall that a unital C*-algebra $A$ admits the local lifting property (LLP) if for any unital C*-algebra B with a closed
two-sided ideal $I\subset B$ any restriction $\mathcal{E}|_\mathcal{S}$ of a unital completely positive map $\mathcal{E} :A\rightarrow B/I$ to arbitrary finite-dimensional operator system $\mathcal{S}\subset A$ lifts to the unital
completely positive map $\tilde{\mathcal{E}}|_\mathcal{S} :\mathcal{S}\rightarrow B$, i.e. $\mathcal{E}|_\mathcal{S} = q \circ \tilde{\mathcal{E}}|_\mathcal{S}$ where $q : B \rightarrow B/I$ stands for the quotient map \cite{Kirchberg} 
- compare this with a stronger notion of the lifting property (LP) evoked in the proof of Proposition \ref{WEP}.

Let $G$ be a discrete group such that $\mathbb{F}_2\subset G$ and  $C^*(G)$ has the local lifting property. According to Proposition 3.8 in \cite{Farenick}, Tsirelson's conjecture is true if and only if $C^*(G)$ has the WEP (by a reasoning similar to the one in Proposition \ref{WEP}). Therefore, taking into account Proposition \ref{WEP} and Theorem \ref{mat_stab} we obtain an immediate corollary.

\begin{corollary}\label{last_Tsirelson}Let $G$ be a discrete group such that $\mathbb{F}_2\subset G$ and $C^*(G)$ admits a local lifting property. Tsirelson's conjecture is false if and only if there exists $n\in \mathbb{N}$ such that $\mathbf{qA}(m,k,M_n(C^*(G)))\neq \mathbf{qcA}(m,k,M_n(C^*(G)))$ for some fixed $(m,k)\neq (2,2)$.
\end{corollary}

\color{black}

This observation provides a possible road towards alternative refutation of Tsirelson's conjecture though a construction of a steering inequality separating quantum commuting and tensor models for some $M_n(C^*(G))$ with $G$ as above. Note that beside discussed examples $*_{i=1}^m\mathbb{Z}_k$ (for $(m,k)\neq(2,2)$) also $SL_2(\mathbb{Z})$ fulfills assumptions of Corollary \ref{last_Tsirelson} \cite{Farenick}.

\section*{Acknowledgements}
This work was initially supported by the Foundation for Polish Science (IRAP project, ICTQT, contract no. MAB/2018/5, co-financed by EU within Smart Growth Operational Programme) and then by the National Science Centre, Poland through grant Maestro (2021/42/A/ST2/00356). The author acknowledge the support of Germany’s Excellence Strategy -- Cluster of Excellence Matter and Light for Quantum Computing (ML4Q), EXC 2004/1 (390534769). The author would like to thank Marcin Marciniak and Paweł Horodecki for stimulating discussions. The author is also grateful to Yin Zhi and Alexander Frei for conversation on the preprint.

\end{document}